\documentclass[conference]{IEEEtran} 

\usepackage{cite}

\ifCLASSINFOpdf
  \usepackage[pdftex]{graphicx}
  \usepackage{epstopdf} 
\else
\fi
\usepackage{amsmath}
\usepackage{algorithm}
\usepackage{algorithmic}

  \usepackage[caption=false,font=footnotesize]{subfig}
\usepackage{url}

\usepackage{amsthm}

\newtheorem{theorem}{Theorem}

\newtheoremstyle{mydefinition}
{}
{}
{}
{0pt}
{\bfseries}
{.}
{ }
{\thmname{#1}\thmnumber{ #2}: \thmnote{#3}}

\theoremstyle{mydefinition}

\newtheoremstyle{myremark}
{}
{}
{}
{0pt}
{\bfseries}
{.}
{ }
{\thmname{#1}\thmnumber{ #2}: \thmnote{#3}}
\theoremstyle{myremark}

\newtheoremstyle{remarkshort}
{}
{}
{}
{0pt}
{\bfseries}
{.}
{ }
{\thmname{#1}\thmnumber{ #2}}

\theoremstyle{remarkshort}

\theoremstyle{remarkshort}

\usepackage{amssymb,amsmath}
\usepackage{xspace}
\usepackage{bm}
\usepackage{bbm}
\usepackage{mathtools}

\let\originalleft\left
\let\originalright\right
\renewcommand{\left}{\mathopen{}\mathclose\bgroup\originalleft}
\renewcommand{\right}{\aftergroup\egroup\originalright}

\newcommand{\set}[1]{{\mathcal{#1}}} 
\newcommand{\mat}[1]{{\mathbf{#1}}} 
\renewcommand{\vec}[1]{{\mathbf{#1}}} 
\newcommand{\matfd}[1]{\bm{\mathsf{#1}}} 
\newcommand{\vecfd}[1]{{\bm{\mathsf{#1}}}} 

\newcommand{\parens}[1]{{\left(#1\right)}\xspace}
\newcommand{\brackets}[1]{{\left[#1\right]}\xspace}
\newcommand{\braces}[1]{{\left\{#1\right\}}\xspace}
\newcommand{\bars}[1]{{\left\vert#1\right\vert}\xspace}
\newcommand{\doublebars}[1]{{\left\Vert#1\right\Vert}\xspace}

\renewcommand{\j}{{\mathrm{j}}}
\newcommand{\e}{{\mathrm{e}}}

\newcommand{\complex}{{\mathbb{C}}\xspace}

\newcommand{\epow}[1]{{\e^{#1}}}
\newcommand{\expp}[1]{\mathrm{exp}\parens{#1}}

\newcommand{\logtwo}[1]{\ensuremath{\mathrm{log}_{2}\parens{#1}}}

\newcommand{\cosp}[1]{{\mathrm{cos}\parens{#1}}}

\newcommand{\card}[1]{\bars{#1}}

\newcommand{\setcomplex}{{\complex}}

\newcommand{\setvector}[2]{#1^{#2 \times 1}}

\newcommand{\setvectorcomplex}[1]{\setvector{\setcomplex}{#1}}

\newcommand{\setmatrix}[3]{#1^{#2 \times #3}}

\newcommand{\setmatrixcomplex}[2]{\setmatrix{\setcomplex}{#1}{#2}}

\newcommand{\ctrans}{^{{*}}}

\newcommand{\trace}[1]{{\mathrm{tr}\parens{#1}}}

\newcommand{\entry}[2]{{\brackets{#1}_{#2}}}

\newcommand{\diag}[1]{{\mathrm{diag}\parens{#1}}}

\newcommand{\normp}[2]{\doublebars{#2}_{#1}}

\newcommand{\normtwo}[1]{\normp{2}{#1}}

\newcommand{\normfro}[1]{\normp{\mathrm{F}}{#1}}

\newcommand{\sigmatx}{\sigma_{\mathrm{tx}}}
\newcommand{\sigmarx}{\sigma_{\mathrm{rx}}}

\DeclareMathOperator*{\argmin}{\mathrm{argmin}}

\newcommand{\subjectto}{\mathrm{s.t.~}}

\newcommand{\opt}{^{\star}}
\newcommand{\project}[2]{\ensuremath{\Pi_{#1}\parens{#2}}\xspace}

\newcommand{\fc}{f_\mathrm{c}}

\newcommand{\powertx}{P_{\mathrm{tx}}}

\newcommand{\powernoise}{P_{\mathrm{n}}}

\newcommand{\snr}{{\mathsf{SNR}}}

\newcommand{\sinr}{{\mathsf{SINR}}}

\newcommand{\inr}{{\mathsf{INR}}}

\newcommand{\snrtxbar}{\overline{\snr}_{\labeltx}}
\newcommand{\snrrxbar}{\overline{\snr}_{\labelrx}}

\newcommand{\inrrxbar}{\overline{\inr}_{\labelrx}}

\newcommand{\sigmatxsq}{\sigmatx^2}
\newcommand{\sigmarxsq}{\sigmarx^2}

\newcommand{\bitsamp}{b_{\mathrm{amp}}}
\newcommand{\bitsphase}{b_{\mathrm{phs}}}

\newcommand{\mAtx}{\mA_{\mathrm{tx}}}
\newcommand{\mArx}{\mA_{\mathrm{rx}}}

\newcommand{\Nt}{N_\mathrm{t}} 
\newcommand{\Nr}{N_\mathrm{r}} 

\newcommand{\precbbar}{{\bar{\mathcal{F}}}\xspace}

\newcommand{\comcbbar}{{\bar{\mathcal{W}}}\xspace}

\newcommand{\labeltx}{\mathrm{tx}}
\newcommand{\labelrx}{\mathrm{rx}}

\newcommand{\snrtx}{\snr_{\labeltx}}
\newcommand{\snrrx}{\snr_{\labelrx}}

\newcommand{\sinrtx}{\sinr_{\labeltx}}
\newcommand{\sinrrx}{\sinr_{\labelrx}}

\newcommand{\inrtx}{\inr_{\labeltx}}
\newcommand{\inrrx}{\inr_{\labelrx}}

\newcommand{\thetatx}{\ensuremath{\theta_{\mathrm{tx}}}\xspace}

\newcommand{\thetarx}{\ensuremath{\theta_{\mathrm{rx}}}\xspace}

\newcommand{\idx}[1]{^{\parens{#1}}}

\def\va{{\vec{a}}}

\def\vf{{\vec{f}}}

\def\vw{{\vec{w}}}
\def\vx{{\vec{x}}}

\def\vone{{\vec{1}}}

\def\vfh{{\vecfd{h}}}

\def\mA{{\mat{A}}}

\def\mF{{\mat{F}}}

\def\mQ{{\mat{Q}}}

\def\mW{{\mat{W}}}

\def\mLambda{{\mat{\Lambda}}}

\def\mfH{{\matfd{H}}}

\def\sA{{\set{A}}}

\def\sF{{\set{F}}}

\def\sP{{\set{P}}}

\def\sW{{\set{W}}}

\newcommand{\vatx}{\va_{\labeltx}\xspace}
\newcommand{\varx}{\va_{\labelrx}\xspace}

\newcommand{\Gtx}{G_{\labeltx}\xspace}
\newcommand{\Grx}{G_{\labelrx}\xspace}

\newcommand{\Mtx}{M_{\labeltx}\xspace}
\newcommand{\Mrx}{M_{\labelrx}\xspace}

\newcommand{\vfhtx}{\vecfd{h}_{\labeltx}\xspace}
\newcommand{\vfhrx}{\vecfd{h}_{\labelrx}\xspace}

\newcommand\eqa{\stackrel{\mathclap{\textrm{(a)}}}{\mathrm{=}}}
\newcommand\eqb{\stackrel{\mathclap{\textrm{(b)}}}{\mathrm{=}}}
\newcommand\eqc{\stackrel{\mathclap{\textrm{(c)}}}{\mathrm{=}}}
\newcommand\eqd{\stackrel{\mathclap{\textrm{(d)}}}{\mathrm{=}}}

\usepackage{xspace}
\usepackage[acronym,nogroupskip,nonumberlist,nopostdot]{glossaries}
\makeglossaries

\newcommand{\lonestar}{\textsc{LoneSTAR}\xspace}

\usepackage{xcolor}

\newacronym{snr}{SNR}{signal-to-noise ratio}
\newacronym{sinr}{SINR}{signal-to-interference-plus-noise ratio}
\newacronym{inr}{INR}{interference-to-noise ratio}
\newacronym{sir}{SIR}{signal-to-interference ratio}
\newacronym{sqr}{SQR}{signal-to-quantization-noise ratio}
\newacronym{sqnr}{SQNR}{signal-to-quantization-plus-noise ratio}
\newacronym{ian}{IAN}{interference as noise}
\newacronym{ber}{BER}{bit error rate}
\newacronym{pn}{PN}{pseudorandom noise}
\newacronym{bfsk}{BFSK}{binary frequency shift keying}
\newacronym{fh}{FH}{frequency-hopped}
\newacronym{fh-bfsk}{FH-BFSK}{frequency-hopped binary frequency shift keying}
\newacronym{crc}{CRC}{cyclic redundancy check}
\newacronym{isi}{ISI}{intersymbol interference}
\newacronym{dsss}{DSSS}{direct-sequence spread spectrum}
\newacronym{ofdm}{OFDM}{orthogonal frequency-division multiplexing}
\newacronym{ofdma}{OFDMA}{orthogonal frequency-division multiple access}
\newacronym{sdr}{SDR}{software-defined radio}
\newacronym{tx}{TX}{transmitter}
\newacronym{rx}{RX}{receiver}
\newacronym{fdd}{FDD}{frequency-division duplexing}
\newacronym{tdd}{TDD}{time-division duplexing}
\newacronym{fdma}{FDMA}{frequency-division multiple access}
\newacronym{tdma}{TDMA}{time-division multiple access}
\newacronym{sdma}{SDMA}{space-division multiple access}
\newacronym[plural=MPCs]{mpc}{MPC}{multipath component}
\newacronym{mui}{MUI}{multi-user interference}
\newacronym{lsb}{LSB}{least significant bit}
\newacronym{jcas}{JCAS}{joint communication and sensing}

\newacronym{qam}{QAM}{quadrature amplitude modulation}
\newacronym{mqam}{MQAM}{M-ary quadrature amplitude modulation}

\newacronym{ls}{LS}{least-squares}
\newacronym{lms}{LMS}{least mean squares}
\newacronym{rls}{RLS}{recursive least-squares}
\newacronym{rzf}{RZF}{regularized zero-forcing}
\newacronym{mmse}{MMSE}{minimum mean square error}
\newacronym{lmmse}{LMMSE}{linear minimum mean square error}
\newacronym{mse}{MSE}{mean square error}
\newacronym{fft}{FFT}{fast Fourier transform}
\newacronym{dft}{DFT}{discrete Fourier transform}
\newacronym{dtft}{DTFT}{discrete-time Fourier transform}
\newacronym{ctft}{CTFT}{continuous-time Fourier transform}
\newacronym{ml}{ML}{machine learning}
\newacronym[plural=NNs]{nn}{NN}{neural network}
\newacronym[plural=RNNs]{rnn}{RNN}{recurrent neural network}
\newacronym[plural=ADCs]{adc}{ADC}{analog-to-digital converter}
\newacronym[plural=DACs]{dac}{DAC}{digital-to-analog converter}
\newacronym[plural=FPGAs]{fpga}{FPGA}{field-programmable gate array}
\newacronym{evm}{EVM}{error vector magnitude}
\newacronym{enob}{ENOB}{effective number of bits}
\newacronym{zf}{ZF}{zero-forcing}
\newacronym{rv}{r.v.}{random variable}
\newacronym{omp}{OMP}{orthogonal matching pursuit}
\newacronym{svd}{SVD}{singular value decomposition}

\newacronym{sdp}{SDP}{semidefinite programming}
\newacronym{psd}{PSD}{positive semidefinite}
\newacronym{nsd}{NSD}{negative semidefinite}

\newacronym{ks}{K-S}{Kolmogorov-Smirnov}

\newacronym{mad}{MAD}{median absolute deviation around the median}

\newacronym{agc}{AGC}{automatic gain control}
\newacronym{rf}{RF}{radio frequency}
\newacronym{if}{IF}{intermediate frequency}
\newacronym{los}{LOS}{line-of-sight}
\newacronym{nlos}{NLOS}{non-line-of-sight}
\newacronym{ple}{PLE}{path loss exponent}
\newacronym[plural=dB,firstplural=decibels (dB)]{db}{dB}{decibel}
\newacronym[plural=dBm,firstplural=decibel milliwatts (dBm)]{dbm}{dBm}{decibel milliwatts}
\newacronym{pa}{PA}{power amplifier}
\newacronym{lna}{LNA}{low noise amplifier}
\newacronym{vga}{VGA}{variable gain amplifier}
\newacronym{cw}{CW}{continuous wave}
\newacronym{papr}{PAPR}{peak-to-average power ratio}
\newacronym{usrp}{USRP}{Universal Software Radio Peripheral}
\newacronym{irr}{IRR}{image rejection ratio}
\newacronym{lo}{LO}{local oscillator}
\newacronym{vm}{VM}{vector modulator}
\newacronym{mmwave}{mmWave}{millimeter-wave}
\newacronym{eirp}{EIRP}{effective isotropic radiated power}
\newacronym{rsrp}{RSRP}{reference signal received power}

\newacronym{csma}{CSMA}{carrier-sense multiple access}
\newacronym{csmaca}{CSMA/CA}{carrier-sense multiple access with collision avoidance}
\newacronym{csmacd}{CSMA/CD}{carrier-sense multiple access with collision detection}
\newacronym{mac}{MAC}{medium access control}
\newacronym{phy}{PHY}{physical layer}
\newacronym{4g}{4G}{fourth generation}
\newacronym{lte}{LTE}{Long-Term Evolution}
\newacronym{4glte}{4G LTE}{\gls{4g} \gls{lte}}
\newacronym{5g}{5G}{fifth generation}
\newacronym{nr}{NR}{New Radio}
\newacronym{5gnr}{5G NR}{5G New Radio}
\newacronym{ieee}{IEEE}{Institute of Electrical and Electronics Engineers}
\newacronym{wifi}{Wi-Fi}{IEEE 802.11}
\newacronym{lan}{LAN}{local area network}
\newacronym{wlan}{WLAN}{wireless local area network}
\newacronym[plural=BSs]{bs}{BS}{base station}
\newacronym[plural=SBSs]{sbs}{SBS}{small-cell base station}
\newacronym[plural=FD-SBSs]{fdsbs}{FD-SBS}{\gls{fd}-enabled \gls{sbs}}
\newacronym[plural=MBSs]{mbs}{MBS}{macrocell base station}
\newacronym[plural=UEs]{ue}{UE}{user equipment}
\newacronym{ul}{UL}{uplink}
\newacronym{dl}{DL}{downlink}
\newacronym{qos}{QoS}{Quality of Service}
\newacronym{fcc}{FCC}{Federal Communications Commission}
\newacronym{iab}{IAB}{integrated access and backhaul}
\newacronym{fab}{FAB}{fixed access and backhaul}
\newacronym{hetnet}{HetNet}{heterogeneous network}

\newacronym{siso}{SISO}{single-input single-output}
\newacronym{mimo}{MIMO}{multiple-input multiple-output}
\newacronym{sumimo}{SU-MIMO}{single-user \gls{mimo}}
\newacronym{mumimo}{MU-MIMO}{multi-user \gls{mimo}}
\newacronym{bf}{BF}{beamforming}
\newacronym{ca}{CA}{constant amplitude}
\newacronym{ula}{ULA}{uniform linear array}
\newacronym{upa}{UPA}{uniform planar array}
\newacronym[\glslongpluralkey={angles of arrival}]{aoa}{AoA}{angle of arrival}
\newacronym[\glslongpluralkey={angles of departure}]{aod}{AoD}{angle of departure}
\newacronym{dof}{DoF}{degrees of freedom}
\newacronym{csi}{CSI}{channel state information}
\newacronym{csit}{CSIT}{\gls{csi} at the transmitter}
\newacronym{csir}{CSIR}{\gls{csi} at the receiver}
\newacronym{cs}{CS}{compressed sensing}

\newacronym{fd}{FD}{in-band full-duplex}
\newacronym{hd}{HD}{half-duplex}
\newacronym{si}{SI}{self-interference}
\newacronym{sic}{SIC}{self-interference cancellation}
\newacronym{soi}{SoI}{signal of interest}
\newacronym{asic}{A-SIC}{analog \acrlong{sic}}
\newacronym{dsic}{D-SIC}{digital \gls{sic}}
\newacronym{star}{STAR}{simultaneous transmit and receive}
\newacronym{warp}{WARP}{Wireless Open-Access Research Platform}
\newacronym{bfc}{BFC}{beamforming cancellation}
\newacronym{ipi}{IPI}{inter-panel-interference}
\newacronym{ipic}{IPIC}{inter-panel-interference cancellation}

\newacronym{qcqp}{QCQP}{quadratically-constrained quadratic programming}
\newacronym{pdf}{PDF}{probability density function}
\newacronym{cdf}{CDF}{cumulative density function}
\newacronym{iid}{i.i.d.}{independently and identically distributed}

\newacronym{elf}{ELF}{extremely low frequency}
\newacronym{slf}{SLF}{super low frequency}
\newacronym{ulf}{ULF}{ultra low frequency}
\newacronym{vlf}{VLF}{very low frequency}
\newacronym{lf}{LF}{low frequency}
\newacronym{mf}{MF}{medium frequency}
\newacronym{hf}{HF}{high frequency}
\newacronym{vhf}{VHF}{very high frequency}
\newacronym{uhf}{UHF}{ultra high frequency}
\newacronym{shf}{SHF}{super high frequency}
\newacronym{ehf}{EHF}{extremely high frequency}
\newacronym{thf}{THF}{tremendously high frequency}

\newacronym{wncg}{WNCG}{Wireless Networking and Communications Group}
\newacronym{linc}{LINC}{Laboratory of Informatics, Networks, and Communications}
\newacronym{ut}{UT Austin}{The University of Texas at Austin}
\newacronym{uiuc}{UIUC}{University of Illinois at Urbana-Champaign}
\newacronym{usc}{USC}{University of Southern California}
\newacronym{mit}{MIT}{Massachusetts Institute of Technology}
\newacronym{berkeley}{UC Berkeley}{University of California, Berkeley}
\newacronym{osu}{OSU}{Ohio State University}

\newcommand{\mmwave}{\gls{mmwave}\xspace}
\newcommand{\mimo}{\gls{mimo}\xspace}

\newcommand{\sic}{\acrlong{sic}\xspace}

\newcommand{\bs}{\gls{bs}\xspace}
\newcommand{\bss}{\glspl{bs}\xspace}
\newcommand{\gsnr}{\gls{snr}\xspace}

\newcommand{\figref}[1]{\figurename~\ref{#1}}

\newcommand{\thmref}[1]{Theorem~\ref{#1}}

\usepackage{mathtools}
\usepackage{comment}

\usepackage{enumitem}

\usepackage{pgfplots}
\pgfplotsset{compat=newest}
\usetikzlibrary{plotmarks}
\usetikzlibrary{arrows.meta}
\usepgfplotslibrary{patchplots}
\usepackage{grffile}

\pgfplotsset{plot coordinates/math parser=false}

\definecolor{uclablue}{RGB}{39,116,174}
\definecolor{uclabluedarkest}{RGB}{0,59,92}
\definecolor{uclabluedarker}{RGB}{0,85,135}
\definecolor{uclabluelighter}{RGB}{139,184,232}
\definecolor{uclabluelightest}{RGB}{218,235,254}

\definecolor{uclagold}{RGB}{255,209,0}
\definecolor{uclagolddarker}{RGB}{255,199,44}
\definecolor{uclagolddarkest}{RGB}{255,184,28}

\definecolor{uclaviolet}{RGB}{130,55,165}
\definecolor{uclamagenta}{RGB}{255,0,165}

\begin{document}

%
\title{Analog Beamforming Codebooks for\\Wideband Full-Duplex Millimeter-Wave Systems}

%
%
%

\author{
\IEEEauthorblockN{Sungho Cho and Ian P.~Roberts}%
\IEEEauthorblockA{Wireless Lab, Department of Electrical and Computer Engineering, University of California, Los Angeles (UCLA)}%
\IEEEauthorblockA{Email: \{shcho1304, ianroberts\}@ucla.edu}
}

\maketitle

\begin{abstract}
In full-duplex \mmwave systems, the effects of beam squint and the frequency-selectivity of self-interference exacerbate over wide bandwidths.
This complicates the use of beamforming to cancel self-interference when communicating over bandwidths on the order of gigahertz. 
In this work, we present the first analog beamforming codebooks tailored to wideband full-duplex \mmwave systems, designed to both combat beam squint and cancel frequency-selective self-interference.
Our proposed design constructs such codebooks by minimizing self-interference across the entire band of interest while constraining the coverage provided by these codebooks across that same band.
Simulation results using computational electromagnetics to model self-interference suggest that a full-duplex 60 GHz system with our design enjoys lower self-interference and delivers better coverage across bandwidths as wide as 6 GHz, when compared to similar codebook designs that ignore beam squint and/or frequency-selectivity. 
This allows our design to sustain higher SINRs and spectral efficiencies across wide bandwidths, unlocking the potentials of wideband full-duplex \mmwave systems. 
\end{abstract}

\glsresetall

\section{Introduction} \label{sec:introduction}
\Gls{mmwave} \bss rely on analog beamforming with dense antenna arrays to transmit and receive signals with high gain in the direction of users they wish to serve \cite{heath_overview_2016}. 
Unlike their low-frequency counterparts, \mmwave systems typically employ beam alignment procedures, where beams are selected from some predefined beam codebook following a series of received power measurements between a \bs and its users \cite{heath_overview_2016, ethan_beam}.
This process circumvents the need to estimate each user's \mimo channel, which is both challenging and resource-expensive, given the channels can be high-dimensional.

In pursuit of enhancing \mmwave networks, there has emerged the potential for a \mmwave \bs to operate in an in-band full-duplex fashion by crafting its transmit and receive beams such that they couple low self-interference while delivering high gain toward a downlink user and uplink user \cite{smida_fd_6g_jsac_2023,roberts_wcm}.
A variety of work has developed such full-duplex beam designs, taking a multitude of approaches ranging from user-specific analog/hybrid beamforming \cite{prelcic_2019_hybrid,roberts_bflrdr,lopez_analog_2022}, beam codebook designs \cite{roberts_lonestar,Bayrak_code_2023}, beam selection \cite{roberts_steer}, and beam refinement \cite{roberts_steerp}. 
These solutions have all highlighted the potential of realizing full-duplex \mmwave \bss without conventional analog or digital \sic \cite{smida_fd_phy_2024}.

The principal motivation for operating at \mmwave carrier frequencies is to leverage wide bandwidths on the order of hundreds of MHz or even several GHz to deliver high data rates.
Until now, however, existing beamforming solutions for full-duplex \mmwave \bss have been almost exclusively narrowband in their design, and thus ignore two noteworthy factors faced by wideband full-duplex \mmwave systems.
The first of these factors is so-called \textit{beam squint} under analog beamforming, where beamforming gain varies with frequency across wide bandwidths due to the array response itself being a function of frequency \cite{phased_array_handbook}.
Ignoring beam squint can thus lead to poor gain across the operating bandwidth, most severely at its edges \cite{phased_array_handbook}.
The second important factor often overlooked is the frequency-selectivity of the self-interference channel, which naturally worsens over wide bandwidths \cite{smida_fd_phy_2024,lee_2015,rf_wide_si_cancel_phased_2022,Joint_analog_dig_wide_cancel_2022}.
Because of this, transmit and receive beams designed for full-duplex using existing narrowband methods do not necessarily provide robust self-interference cancellation across the entire operating bandwidth but rather only at and around the carrier frequency.
To the best of our knowledge, there is no existing work which designs analog beamformers specifically for wideband full-duplex \mmwave systems, taking into account both beam squint and wideband self-interference cancellation.

In this work, rather than solely focus on designing a single pair of transmit and receive beams, we design entire beam \textit{codebooks} for wideband full-duplex \mmwave systems.
While this is more challenging, we set our sights on such codebooks as this will allow real-world full-duplex \mmwave systems (such as those in 5G/6G) to accommodate beam alignment procedures, an approach also taken in \cite{roberts_lonestar} and similarly in \cite{Bayrak_code_2023}.
We design our proposed codebooks for wideband full-duplex \mmwave systems using convex optimization to ensure that the beam gain delivered by the codebooks is constrained to a certain quality across frequency while minimizing self-interference over the entire bandwidth. 
As a result, our proposed full-duplex codebooks are capable of offering higher spectral efficiencies over wide bandwidths than existing narrowband codebooks.

\begin{figure}
    \centering%
    \includegraphics[width=\textwidth,height=0.29\textheight,keepaspectratio]{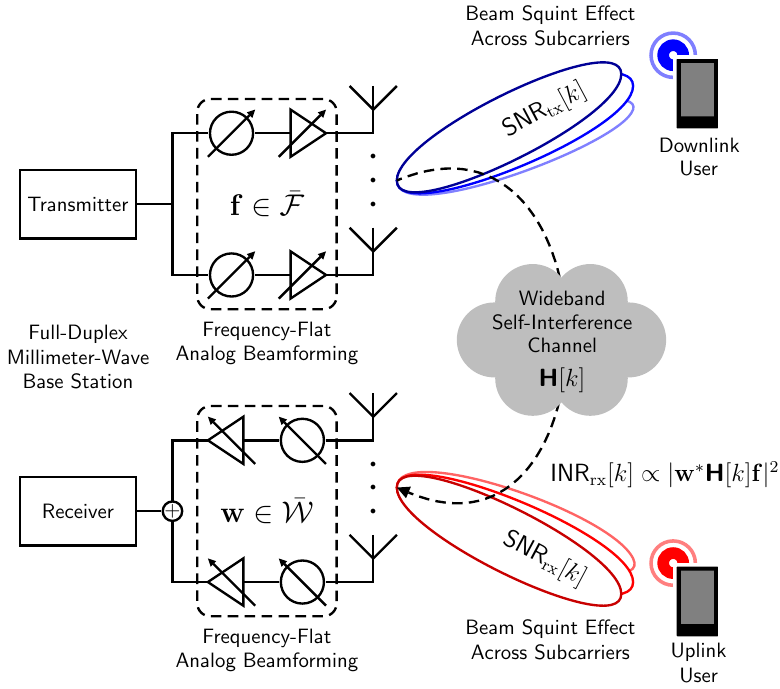}%
    \caption{A wideband full-duplex \mmwave \bs serves downlink and uplink users at the same time and frequency. The self-interference channel and the delivered SNR varies across OFDM subcarriers (indexed by $k$), the latter due in part to the beam squint effect seen with analog beamforming.}
    \label{fig:system}%
\end{figure}

\section{System Model} \label{sec:system-model}
This work considers a wideband \mmwave \bs serving a downlink user and an uplink user at the same time and same frequency (i.e., in an in-band full-duplex fashion).
The \bs operates at a carrier frequency $\fc$ on the order of 30--100~GHz and over a bandwidth $B$ spanning hundreds of MHz or multiple GHz.
Each user is assumed to have a single omnidirectional antenna, but this is not a necessary assumption.
To serve the users, the \bs is equipped with separate transmit and receive antenna arrays, containing $\Nt$ and $\Nr$ antennas, respectively.
Let $\vatx(\theta,f) \in \setvectorcomplex{\Nt}$ and $\varx(\theta,f) \in \setvectorcomplex{\Nr}$ be the array response vectors of the transmit and receive antenna arrays in some direction $\theta$ and at some frequency $f$.
For instance, when $\theta$ is comprised of azimuth $\vartheta$ and elevation $\varphi$ components, the response of the $i$-th element of the transmit array at frequency $f$ is of the form
\begin{align}
    [\vatx(\theta, f)]_i 
    = \expp{\j \, 2\pi \, \frac{f}{\fc} \, \Phi_i(\vartheta,\varphi)},
\end{align}
with $\Phi_i(\vartheta,\varphi) = x_i \sin \vartheta \cos \varphi + y_i \cos\vartheta \cos \varphi + z_i \sin \varphi$, where $(x_i,y_i,z_i)$ is the Cartesian coordinate of the $i$-th transmit antenna. 
In the context of this work, it is important to recognize that the responses of the transmit and receive  arrays depend on frequency~$f$; this leads to the well-known beam squint effect in wideband beamforming systems \cite{phased_array_handbook}.

The transmit and receive arrays are steered electronically using beamforming weights $\vf\in\setvectorcomplex{\Nt}$ and $\vw\in\setvectorcomplex{\Nr}$. 
In the majority of practical \mmwave phased arrays, analog beamforming is employed, and we thus assume digitally controlled phase shifters and stepped attenuators are used to physically realize desired transmit and receive beamforming weights $\vf$ and $\vw$ in hardware.
Let $\sP$ and $\sA$ be the discrete sets of physically realizable phase and amplitude settings at any given antenna, respectively, where $\card{\sP} = 2^{\bitsphase}$ and $\card{\sA} = 2^{\bitsamp}$, with $\bitsphase$ and $\bitsamp$ the resolution (in bits) of each phase shifter and attenuator. 
With this, the sets of all physically realizable transmit beams $\sF$ and receive beams $\sW$ are discrete sets defined as 
\begin{align}
\sF &= \braces{\vf : \entry{\vf}{m} = A \epow{\j\theta}, \theta \in \sP, A \in \sA, m = 1,\dots,\Nt} \\
\sW &= \braces{\vw : \entry{\vw}{n} = A \epow{\j\theta}, \theta \in \sP, A \in \sA, n = 1,\dots,\Nr},\! 
\end{align}
and it follows that $\vf \in \sF \subset \setvectorcomplex{\Nt}$ and $\vw \in \sW \subset \setvectorcomplex{\Nr}$.
Employing attenuators, each beamforming weight must not exceed unit magnitude, and thus $\normtwo{\vf}^2 \leq \Nt$ and $\normtwo{\vw}^2 \leq \Nr$. 

We assume that downlink and uplink employ \gls{ofdm} with $K$ subcarriers; given our focus is full-duplex operation, we assume downlink and uplink overlap on all $K$ subcarriers.
Let $\powertx^{\mathsf{BS}}$ and $\powertx^{\mathsf{UE}}$ be the total transmit powers of the \bs and uplink user, which we assume are allocated uniformly across subcarriers for simplicity, though this is not a necessary assumption.
Let $\powernoise^{\mathsf{BS}}$ and $\powernoise^{\mathsf{UE}}$ be the integrated noise powers of the \bs and downlink user, with the noise power spectral densities assumed flat across subcarriers.

Let us denote the downlink channel on the $k$-th subcarrier by $\vfhtx[k] \in \setvectorcomplex{\Nt}$ and that of the uplink by $\vfhrx[k] \in \setvectorcomplex{\Nr}$.
We normalize the downlink and uplink channel energies as $\frac{1}{K} \sum_{k=1}^K \normtwo{\vfhtx[k]}^2 = \Nt$ and $\frac{1}{K} \sum_{k=1}^K \normtwo{\vfhrx[k]}^2 = \Nr$, abstracting our their inverse path losses by $\Gtx$ and $\Grx$, respectively.
By operating in a full-duplex fashion, the \bs sees on the $k$-th subcarrier a self-interference channel denoted by $\mfH[k] \in \setmatrixcomplex{\Nr}{\Nt}$.
We normalize the self-interference channel as $\frac{1}{K} \sum_{k=1}^K \normfro{\mfH[k]}^2 = \Nt\Nr$, and account for its inverse path loss by $G$.
Putting this all together, the \glspl{sinr} on the $k$-th subcarrier of the downlink and uplink can be expressed as
\begin{align}
    \sinrtx[k] = \frac{\snrtx[k]}{1 + \inrtx[k]}, \\
    \sinrrx[k] = \frac{\snrrx[k]}{1 + \inrrx[k]},
\end{align}
with the downlink and uplink \glspl{snr} on the $k$-th subcarrier defined as
\begin{align}
    \snrtx[k] &= \frac{\powertx^{\mathsf{BS}} \, \Gtx \, \bars{\vfh_{\mathrm{tx}}^*[k] \vf}^2}{\Nt \, \powernoise^{\mathsf{UE}}} \label{eq:snrtx}, \\
    \snrrx[k] &= \frac{\powertx^{\mathsf{UE}} \, \Grx \, \bars{\vfh_{\mathrm{rx}}\ctrans[k] \vw}^2}{\normtwo{\vw}^2 \, \powernoise^{\mathsf{BS}}}, \label{eq:snrrx}
\end{align}
where dividing by $\Nt$ in \eqref{eq:snrtx} accounts for signal splitting in the transmit array and by $\normtwo{\vw}^2$ in \eqref{eq:snrrx} accounts for noise combining in the receive array. 

Though not explicitly formulated, the \gls{inr} on the downlink, denoted by $\inrtx[k]$, is due to cross-link interference between the single-antenna users. 
More relevant in this work is the \gls{inr} on the uplink, which stems from the transmit and receive beams coupling across the self-interference channel.
This uplink \gls{inr} on the $k$-th subcarrier can be expressed as
\begin{align} \label{inrrx}
\inrrx[k] = \frac{\powertx^{\mathsf{BS}} \, G \, \bars{\vw\ctrans \mfH[k] \vf}^2}{\Nt \, \normtwo{\vw}^2 \, \powernoise^{\mathsf{BS}}}.
\end{align}
It will be useful to denote the upper bounds on the average \gls{snr} across subcarriers on the downlink and uplink by $\snrtxbar$ and $\snrrxbar$, which can be written explicitly as
\begin{align}
    \snrtxbar = \frac{\powertx^{\mathsf{BS}} \, \Gtx \, \Nt}{\powernoise^{\mathsf{UE}}} \geq \frac{1}{K}\sum_{k=1}^K\snrtx[k], \\
    \snrrxbar = \frac{\powertx^{\mathsf{UE}} \, \Grx \, \Nr}{\powernoise^{\mathsf{BS}}} \geq \frac{1}{K}\sum_{k=1}^K\snrrx[k].
\end{align}
These upper bounds are a consequence of the normalizations on the downlink and uplink channels.
Likewise, the upper bound on the average of $\inrrx[k]$ can be defined as
\begin{align}
    \inrrxbar = \frac{\powertx^{\mathsf{BS}} \, G \, \Nt \, \Nr}{\powernoise^{\mathsf{BS}}} \geq \frac{1}{K}\sum_{k=1}^K \inrrx[k].
\end{align}
Notice that $\snrtxbar$, $\snrrxbar$, and $\inrrxbar$ only depend on system parameters and thus represent the inherent link quality on the downlink and uplink and severity of self-interference.
Treating interference as noise, the achievable sum spectral efficiency on the $k$-th subcarrier of the downlink and uplink is
\begin{align}
    R[k] &= \logtwo{1 + \sinrtx[k]} + \logtwo{1 + \sinrrx[k]}.
\end{align}
In the next section, we aim to maximize the sum spectral efficiency across subcarriers $\frac{1}{K} \sum_{k=1}^K R[k]$.

\section{Wideband Full-Duplex Beam Codebook Design} \label{sec:contribution}
Practical \mmwave wireless systems, such as those in 5G and WiGig, rely on beam alignment procedures to close the link between a \bs and the users it intends to serve \cite{heath_overview_2016,ethan_beam}.
These procedures involve over-the-air measurements with candidate/probing beams, followed by beam selection from some predefined codebook of serving beams for data transmission/reception.
This process circumvents the need for high-dimensional \gls{mimo} channel acquisition, which is resource-expensive and challenging in practical \mmwave systems.

Under the assumption that the considered full-duplex \mmwave system will employ such beam alignment procedures, we narrow the focus of this work to designing predefined codebooks of transmit and receive beams for downlink and uplink service.
Let $\precbbar$ and $\comcbbar$ denote codebooks of $\Mtx$ transmit and $\Mrx$ receive beams as
\begin{align}
\precbbar & = \braces{\vf_1, \vf_2, \dots, \vf_{\Mtx}} \subset \sF \subset \setvectorcomplex{\Nt} \\
\comcbbar & = \braces{\vw_1, \vw_2, \dots, \vw_{\Mrx}} \subset \sW \subset \setvectorcomplex{\Nr}.
\end{align}
We assume that the $i$-th transmit beam $\vf_i$ and the $j$-th receive beam $\vw_j$ are intended to steer with high gain toward some directions $\thetatx\idx{i}$ and $\thetarx\idx{j}$, respectively.
The coverage regions spanned by $\precbbar$ and $\comcbbar$ can then be loosely defined as $\{\thetatx\idx{1},\dots,\thetatx\idx{\Mtx}\}$ and $\{\thetarx\idx{1},\dots,\thetarx\idx{\Mrx}\}$.
The gain delivered by the $i$-th transmit beam $\vf_i$ in its serving direction $\thetatx\idx{i}$ at frequency $f$ can be expressed as $|\vf_i\ctrans \vatx(\thetatx\idx{i},f)|^2$, for instance.
Thus, the gain delivered by a given beam depends on direction and frequency, giving rise to beam squint \cite{phased_array_handbook}.

The aim of this work is to design the transmit and receive codebooks $\precbbar$ and $\comcbbar$ such that they: (i) deliver high beamforming gain on the downlink and uplink across the coverage regions and across all subcarriers $k \in [1,2,\dots,K]$ and (ii) couple low self-interference $\inrrx[k]$ across all subcarriers $k \in [1,2,\dots,K]$ to enable full-duplex operation.
Satisfying these two goals will ensure that any transmit and receive beams from within those codebooks can be reliably used to serve downlink and uplink users in a full-duplex fashion.
To formalize these two aims, on subcarrier $k$ at frequency $f_k$, let the collection of transmit array responses evaluated across the downlink coverage region be denoted by $\mAtx[k] \in \setmatrixcomplex{\Nt}{\Mtx}$, and let $\mArx[k] \in \setmatrixcomplex{\Nr}{\Mrx}$ be defined analogously.
\begin{align}
    \mAtx[k] &= 
    \begin{bmatrix}
    \vatx\parens{\thetatx\idx{1}, f_k} 
    & \cdots &
    \vatx\parens{\thetatx\idx{\Mtx},f_k}
    \end{bmatrix} \label{eq:Atx}
    \\
    \mArx[k] &= 
    \begin{bmatrix}
    \varx\parens{\thetarx\idx{1}, f_k} 
    & \cdots &
    \varx\parens{\thetarx\idx{\Mrx},f_k}
    \end{bmatrix} \label{eq:Arx}
\end{align}
Let us populate the codebooks into matrix form as follows.
\begin{align}
    \mF &= [\vf_1 \quad \vf_2 \quad \cdots \quad \vf_{\Mtx}] \in \setmatrixcomplex{\Nt}{\Mtx} \\
    \mW &= [\vw_1 \quad \vw_2 \quad \cdots \quad \vw_{\Mrx}] \in \setmatrixcomplex{\Nr}{\Mrx}
\end{align}
Then, a constraint on the beamforming gain afforded by the codebooks across all $K$ subcarriers can be formulated as
\begin{align}
    \normtwo{\Nt \cdot \vone - \diag{\mAtx\ctrans[k]\mF}}^2 &\leq \sigmatxsq \, \Nt^2 \, \Mtx \ \forall \  k \label{eq:const-F-gain}\\
    \normtwo{\Nr \cdot \vone - \diag{\mArx\ctrans[k]\mW}}^2 &\leq \sigmarxsq \, \Nr^2 \, \Mrx \ \forall \ k. \label{eq:const-W-gain}
\end{align}
Here, $\sigmatxsq,\sigmarxsq \geq 0$ are design parameters which represent the tolerable variance in normalized gain offered by the codebooks across their respective coverage regions at each subcarrier $k$.
Increasing $\sigmatxsq,\sigmarxsq$ will lead to greater variance in gain and a degradation in coverage but, as we will see, will allow us to construct codebooks which are more robust to self-interference.
We thus term \eqref{eq:const-F-gain} and \eqref{eq:const-W-gain} \textit{coverage constraints}.

To unlock wideband full-duplex operation, self-interference should be reduced across all $K$ subcarriers. 
In this pursuit, we focus our attention on minimizing the average of $\inrrx[k]$ across all $K$ subcarriers and across all possible $\Mtx\Mrx$ transmit-receive beam pairs, expressed as
\begin{subequations}
\begin{align}
\inrrx^{\mathrm{avg}} 
&= \frac{1}{K \Mtx \Mrx} \sum_{k=1}^K \sum_{i=1}^{\Mtx} \sum_{j=1}^{\Mrx} \frac{\powertx^{\mathsf{BS}} G \bars{\vw_j^* \mfH[k] \vf_i}^2}{\powernoise^{\mathsf{BS}} \, \Nt \, \normtwo{\vw_j}^2} \\
&\approx \frac{\powertx^{\mathsf{BS}} \, G}{\powernoise^{\mathsf{BS}} \, \Nt \, \Nr \, K} \sum_{k=1}^K \frac{\normfro{\mW^* \mfH[k] \mF}^2}{\Mtx \, \Mrx}, \label{eq:inr-rx-avg}
\end{align}
\end{subequations}
where the approximation follows when $\normtwo{\vw_j}^2 \approx \Nr$ for all $j$. 
From this, it can be seen that minimizing $\sum_{k=1}^K \normfro{\mW^* \mfH[k] \mF}^2$ will approximately minimize the average self-interference across the band of interest.
Combining the coverage constraints \eqref{eq:const-F-gain} and \eqref{eq:const-W-gain} with the goal of minimizing self-interference, we assemble the following codebook design problem.
\begin{subequations} \label{eq:opt-1}
\begin{align}
\min_{\mF, \mW} \ & \sum_{k=1}^{K}\normfro{\mW^*\mfH[k]\mF}^2 \label{eq:opt-1-obj} \\
\subjectto 
&\normtwo{\Nt \cdot \vone - \diag{\mAtx\ctrans[k]\mF}}^2 \leq \sigmatxsq \, \Nt^2 \, \Mtx \, \forall \, k \label{eq:opt-1-F-gain} \\
&\normtwo{\Nr \cdot \vone - \diag{\mArx\ctrans[k]\mW}}^2 \leq \sigmarxsq \, \Nr^2 \, \Mrx \ \forall \ k \label{eq:opt-1-W-gain} \\
&[\mF]_{:,i} \in \mathcal{F} \ \forall \  i=1,\dots,\Mtx \label{eq:opt-1-F-feas} \\
&[\mW]_{:,j} \in \mathcal{W} \ \forall \ j=1,\dots,\Mrx \label{eq:opt-1-W-feas}
\end{align}
\end{subequations}
Here, constraints \eqref{eq:opt-1-F-feas} and \eqref{eq:opt-1-W-feas} account for finite-resolution phase and amplitude control of practical analog beamforming networks.
Inspecting \eqref{eq:opt-1} reveals the challenge of the problem at hand.
Frequency-flat beam codebooks $\mF$ and $\mW$ are tasked with jointly 
(i) reducing their coupling over a wideband, frequency-selective self-interference channel and 
(ii) delivering high gain across a wide bandwidth, subject to beam squint. 

The objective \eqref{eq:opt-1-obj}, being the sum of $K$ Frobenius norms, can lead to unnecessary computational complexity if used directly since the number of subcarriers $K$ can be large in wideband systems.
In light of this, we introduce the following theorem to simplify the objective into a single Frobenius norm.
\begin{theorem}\label{thm:rewrite-obj}
    For a given receive codebook $\mW$, the objective of \eqref{eq:opt-1-obj} can be written equivalently as
    \begin{align}
        \sum_{k=1}^{K}\normfro{\mW^*\mfH[k]\mF}^2 &= \normfro{\mLambda^{\frac{1}{2}}_{\mW} \mQ^*_{\mW} \mF}^2, \label{eq:equiv-obj}
    \end{align}
    where $\mQ_\mW \mLambda_\mW \mQ_{\mW}\ctrans$ is the eigenvalue decomposition of $ \sum_{k=1}^{K}\mfH^*[k] \mW \mW^* \mfH[k]$.
\end{theorem}
\begin{proof}
This equivalency can be shown as
\begin{subequations}
    \begin{align}
    \sum_{k=1}^{K}\normfro{\mW^*\mfH[k]\mF}^2 
    &\eqa \sum_{k=1}^K\trace{\mF^* \mfH^*[k] \mW \mW^* \mfH[k] \mF} \\
    &\eqb \trace{\mF^* \mQ_{\mW} \mLambda_{\mW} \mQ^*_{\mW} \mF} \\
    &\eqc \trace{\mF^* \mQ_{\mW} \mLambda^{* \frac{1}{2}}_{\mW} \mLambda^{\frac{1}{2}}_{\mW} \mQ^*_{\mW} \mF} \\
    &\eqd \normfro{\mLambda^{\frac{1}{2}}_{\mW} \mQ^*_{\mW} \mF}^2,
\end{align}
\end{subequations}
where 
(a) follows from the definition of the Frobenius norm; 
(b) is by the linearity of trace and by taking the eigenvalue decomposition $\sum_{k=1}^{K}\mfH^*[k] \mW \mW^* \mfH[k] = \mQ_{\mW} \mLambda_{\mW} \mQ^*_{\mW}$;
(c) follows from all eigenvalues being positive; 
and (d) follows from the definition of the Frobenius norm.
\end{proof}

There a few noteworthy challenges associated with solving problem \eqref{eq:opt-1}.
First, the non-convexity of $\sF$ and $\sW$ in constraints \eqref{eq:opt-1-F-feas} and \eqref{eq:opt-1-W-feas} introduces difficulty in optimization.
To address this, we will relax these constraints to continuous ones and then use projection to ensure the designed codebooks are physically realizable, satisfying \eqref{eq:opt-1-F-feas} and \eqref{eq:opt-1-W-feas}.
The second challenge lies in jointly optimizing $\mF$ and $\mW$, which we will overcome by optimizing them in an alternating manner.
We walk through our approach in the following, first defining our projection rule.

\begin{theorem}
Let $\project{\sF}{\vf}$ denote the projection of a beamforming vector $\vf$ onto the set physically realizable beams $\sF$, defined as $\project{\sF}{\vf} = \argmin_{\vx\in\sF} \, \normtwo{\vf - \vx}^2$.
Then, the $m$-th projected weight can be expressed as
\begin{align}
    \entry{\project{\sF}{\vf}}{m} &= A\opt \cdot \expp{\j\theta\opt},
\end{align}
where the projected phase and amplitude are respectively
\begin{align}
\theta\opt 
&= \argmin_{\theta_i \in \sP} \ |\theta_i - \theta|,
\label{opt_theta} \\
A\opt &= \argmin_{A_i \in \sA} \ \bars{A_i - A \cdot \cosp{\theta_i - \theta\opt}}, \label{opt_A}
\end{align}
where the $m$-th entry of $\vf$ is $\entry{\vf}{m} =  A \cdot \expp{\j\theta}$. 
\begin{proof}
The proof is omitted due to space constraints.
\end{proof}
\end{theorem}%
To avoid the complexity associated with jointly optimizing both $\mF$ and $\mW$, we divide solving problem \eqref{eq:opt-1} into two parts. 
We first initialize $\mW$ and solve for $\mF$ and then solve for $\mW$, fixing $\mF$.
This alternating minimization begins by initializing the receive codebook as $\mW \leftarrow \project{\sW}{\mArx[k_0]}$, where $k_0$ is the subcarrier index of the carrier frequency $\fc$ and the projection is applied column-wise to $\mArx[k_0]$.
Then, upon fixing $\mW$, we formulate the following problem to solve for $\mF$, employing the equivalency \eqref{eq:equiv-obj} of \thmref{thm:rewrite-obj} in writing the objective.
\begin{subequations} \label{eq:opt-2}
\begin{align}
\min_{\mF} \
&\normfro{\mLambda^{\frac{1}{2}}_{\mW} \mQ^*_{\mW} \mF}^2 \label{eq:opt-2-obj} \\
\subjectto
&\normtwo{\Nt \cdot \vone - \diag{\mA\ctrans_{\mathrm{tx}}[k]\mF}}^2 \leq \sigmatxsq \, \Nt^2 \, \Mtx \ \forall \  k \label{eq:opt-2-F-gain} \\
&\bars{[\mF]_{i,j}} \leq 1 \ \forall \ i =1,\dots,\Nt, \ j=1,\dots,\Mtx \label{eq:opt-2-F-feas}
\end{align}
\end{subequations}
Here, we have relaxed the non-convex constraint of \eqref{eq:opt-1-F-feas} to the convex constraint \eqref{eq:opt-2-F-feas} since each weight must not exceed unit magnitude, and thus \eqref{eq:opt-2} is a convex problem and can be readily solved to obtain a transmit codebook $\mF\opt$ using convex optimization solvers \cite{cvx}. 
To ensure the final transmit codebook $\mF$ is physically realizable, we project the solution $\mF\opt$ using our projection rule $\mF \leftarrow \project{\sF}{\mF\opt}$.
Having designed $\mF$, we then fix it and solve for the received codebook $\mW$ in a similar manner in the following convex problem.
\begin{subequations} \label{eq:opt-3}
\begin{align}
\min_{\mW} \
&\normfro{\mLambda^{\frac{1}{2}}_{\mF} \mQ^*_{\mF} \mW}^2 \label{eq:opt-3-obj} \\
\subjectto 
&\normtwo{\Nr \cdot \vone - \diag{\mA\ctrans_{\mathrm{rx}}[k]\mW}}^2 \leq \sigmarxsq \, \Nr^2 \, \Mrx \ \forall \ k \label{eq:opt-3-W-gain} \\
&\bars{[\mW]_{i,j}} \leq 1 \ \forall \ i =1,\dots,\Nr, \ j=1,\dots,\Mrx \label{eq:opt-3-W-feas}
\end{align}
\end{subequations}
Here, $\mLambda_{\mF}$ and $\mQ_{\mF}$ are found analogous to those in \thmref{thm:rewrite-obj}.
Upon solving problem \eqref{eq:opt-3} for $\mW\opt$, a physically realizable receive codebook can be obtained by projecting it as $\mW \leftarrow \project{\sW}{\mW\opt}$.
This concludes our design of analog beamforming codebooks for wideband full-duplex \mmwave systems.

\section{Numerical Results} \label{sec:numerical-results}

To assess our proposed wideband codebook design, we simulated the full-duplex system illustrated in \figref{fig:system} at a carrier frequency of 60 GHz across bandwidths of at most 6 GHz. 
The \bs is equipped with two 8$\times$8 half-wavelength uniform planar arrays, whose centers are separated horizontally by 10 carrier wavelengths (about 5 cm).
For analog beamforming, the phase shifters have a resolution of $\bitsphase = 6$ bits, with uniform control across $[0,2\pi)$, and the stepped attenuators also have a resolution of $\bitsamp = 6$ bits and a step size of $0.5$ dB per least significant bit.
The downlink and uplink coverage areas span from $-60^\circ$ to $60^\circ$ in azimuth and from $-30^\circ$ to $30^\circ$ in elevation, with $15^\circ$ spacing in each.
This results in codebooks of $45$ beams for both transmission and reception.
Users are uniformly distributed from $-67.5^\circ$ to $67.5^\circ$ in azimuth and from $-37.5^\circ$ to $37.5^\circ$ in elevation from the \bs, with their channels assumed to be line-of-sight for simplicity.
In line with conventional beam alignment, the transmit and receive beams are selected from their respective codebooks to maximize \gsnr.
The self-interference channel $\braces{\mfH[k]}_{k=1}^K$ was simulated using MATLAB's Antenna Toolbox.
On the downlink and uplink, we set $\snrtxbar=\snrrxbar=10$~dB.
We set $\inrrxbar=80$~dB and ignore cross-link interference.
We compare our proposed design against three codebook designs:
\begin{itemize}
    \item \textbf{Conjugate beamforming (CBF):} 
    Each transmit and receive beam delivers maximum gain in its intended direction at the center frequency. Both self-interference and beam squint are ignored.
    \item \textbf{\lonestar \cite{roberts_lonestar}:} Transmit and receive beams are designed to couple low self-interference and provide high gain in their intended directions at the center frequency. Wideband self-interference and beam squint are ignored.
    \item \textbf{\lonestar-WB:} We replace the narrowband objective of \lonestar \cite{roberts_lonestar} with the wideband objective \eqref{eq:opt-1-obj} proposed in this work. The coverage constraints remain narrowband, ignoring beam squint.
\end{itemize}
Note that \lonestar-WB and the proposed design share the same objective but different constraints; the proposed design ensures a maximum coverage variance across all $K$ subcarriers whereas \lonestar-WB only does so at the center frequency.
For the sake of assessing our proposed scheme, we define the \textit{codebook capacity} as the full-duplex capacity with CBF codebooks if $\inrrxbar=-\infty$ dB.

Let us first consider \figref{fig:inrrx}, depicting the average $\inrrx$ as a function of frequency for CBF, \lonestar \cite{roberts_lonestar}, and our proposed codebooks for various bandwidths $B$. 
Recall that $\inrrx=0$~dB corresponds to self-interference as strong as noise.
For the proposed designs, the coverage variance design parameters $\sigmatxsq=\sigmarxsq$ have been tuned to maximize sum spectral efficiency across the entire bandwidth of interest.
For \lonestar \cite{roberts_lonestar} and \lonestar-WB, $\sigmatxsq=\sigmarxsq$ have been tuned to maximize the sum spectral efficiency across a bandwidth of 6 GHz. 
First, it can be seen that the CBF codebooks couple prohibitively high self-interference across the entire 6 GHz~band.
This is not surprising as their design does not account for self-interference. 
With \lonestar \cite{roberts_lonestar} codebooks, self-interference is markedly reduced to around $15$ dB below noise at the center frequency of 60~GHz, but is well above noise at the edges of the band. 
This confirms that the self-interference channel $\braces{\mfH[k]}$ exhibits significant frequency-selectivity and suggests that narrowband designs are only suitable for bandwidths less than 2~GHz. 
With \lonestar-WB, $\inrrx$ is generally lower than \lonestar \cite{roberts_lonestar}, except around the center frequency, but also exhibits high self-interference near edges of the band.
Our proposed codebooks, on the other hand, reliably reduce self-interference to below noise across the entire band of interest.
It is interesting to note that, while \lonestar-WB is less constrained in general than the proposed design, it exhibits higher $\inrrx$ for the case of $B=6$~GHz than the proposed design.
This is because, as we will see, \lonestar-WB tends to sacrifice coverage across frequency, whereas the proposed design maintains coverage. 
As a result, \lonestar-WB cannot generally attain both low $\inrrx$ and high $\snrtx$ and $\snrrx$, and thus must prioritize \gls{snr} to maximize sum spectral efficiency.

\begin{figure}[t]
    \centering%
    \input{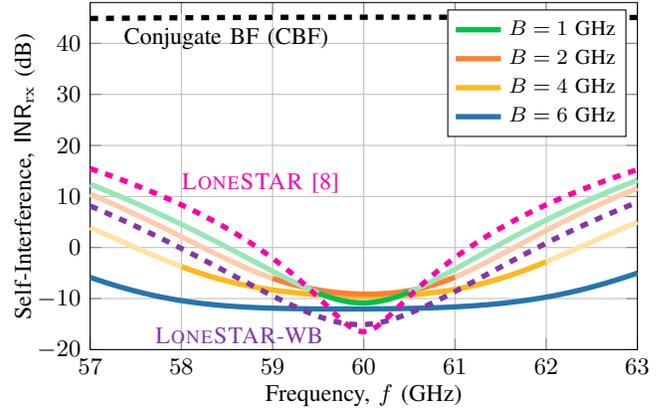}%
    \vspace{-0.25cm}%
    \caption{Average self-interference across frequency for CBF, \lonestar \cite{roberts_lonestar}, \lonestar-WB, and the proposed codebooks for various bandwidths $B$.}
    \label{fig:inrrx}
    \vspace{-0.25cm}
\end{figure}

Having confirmed that our proposed codebooks are indeed capable of canceling wideband self-interference, we now assess the sum spectral efficiency $R=\frac{1}{K}\sum_{k=1}^K R[k]$ they attain in \figref{fig:R}. 
First, we can see that CBF codebooks offer poor spectral efficiency due to the high self-interference they couple.
For \lonestar \cite{roberts_lonestar}, we consider two cases: (i) the dashed line corresponds to fixing $\sigmatxsq=\sigmarxsq$ to that which maximizes the narrowband sum spectral efficiency at the center frequency; and (ii) the solid line corresponds to tuning $\sigmatxsq=\sigmarxsq$ to maximize the sum spectral efficiency across the bandwidth~$B$.
We also include performance for \lonestar-WB, which is also tuned to maximize sum spectral efficiency across $B$.
All three offer satisfactory performance across narrow bandwidths less than 2 GHz, but as the bandwidth is widened, \lonestar \cite{roberts_lonestar} begins to degrade substantially if not tuned to $B$. 
Upon tuning \lonestar, its performance falls just shy of \lonestar-WB, which itself falls short of the proposed design, especially for large bandwidths. 
From this, we can  importantly conclude that \lonestar and \lonestar-WB cannot attain the performance of our proposed design by adjusting $\sigmatxsq=\sigmarxsq$. 
Merely changing \lonestar's objective from narrowband to wideband proves insufficient.
Rather, the objective \textit{and} the coverage constraints should both be formulated in a wideband sense to account for frequency-selectivity and beam squint.
Doing so allows our proposed design to better maintain performance across wide bandwidths, sacrificing only about 10\% of its sum spectral efficiency when the bandwidth increases from 100~MHz to 6~GHz.

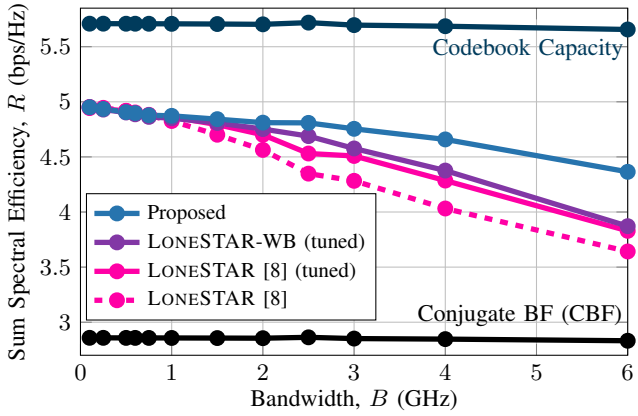
\begin{figure}[t]
    \centering%
%
%
\definecolor{mycolor1}{rgb}{1.00000,0.49412,0.18824}%
\definecolor{mycolor2}{rgb}{0.15294,0.45490,0.68235}%
\begin{tikzpicture}

\begin{axis}[%
width=\linewidth,
height=0.7\linewidth,
xmin=0,
xmax=6,
xlabel style={font=\color{white!0!black}\small},
xlabel={Bandwidth, $B$ (GHz)},
ymin=2.7,
ymax=5.85,
ticklabel style={font=\small},
ylabel style={font=\color{white!0!black}\small},
ylabel={Sum Spectral Efficiency, $R$ (bps/Hz)},
axis background/.style={fill=white},
xmajorgrids,
ymajorgrids,
xtick distance=1,
ytick distance=0.5,
xlabel shift = -3pt,
ylabel shift = -3pt,
legend style={at={(0.01,0.1)}, anchor=south west, legend cell align=left, align=left, draw=white!0!black, legend reversed, font=\footnotesize, inner sep=2pt}
]

\addplot [color=black, line width=2.0pt, mark=*, mark options={solid, black}, forget plot]
  table[row sep=crcr]{%
0.1	2.85837543605935\\
0.25	2.85815932568005\\
0.5	2.85769219537124\\
0.6	2.85766664623853\\
0.75	2.85745159596178\\
1	2.85705387898502\\
1.5	2.85610231891351\\
2	2.85489195306657\\
2.5	2.86280172174239\\
3	2.85147576484338\\
4	2.84654517565208\\
6	2.8317482622804\\
} node[above=0.25mm,pos=0.8] {\small Conjugate BF (CBF)};

\addplot [color=uclamagenta, line width=2.0pt, dashed, mark=*, mark options={solid, uclamagenta}]
table[row sep=crcr]{%
    0.1	4.95324317139319\\
    0.25	4.94550502298169\\
    0.5	4.91800201520155\\
    0.6	4.9032783961111\\
    0.75	4.87670867095431\\
    1	4.82592272871524\\
    1.5	4.70128796058455\\
    2	4.56350933611388\\
    2.5	4.34829519545444\\
    3	4.28309385147094\\
    4	4.03190832647782\\
    6	3.64194011484685\\
};
\addlegendentry{\textsc{LoneSTAR} \cite{roberts_lonestar}}

\addplot [color=uclamagenta, line width=2.0pt, mark=*, mark options={solid, uclamagenta}]
table[row sep=crcr]{%
0.1	4.94296204517546\\
0.25	4.9359535376598\\
0.5	4.91801233887269\\
0.6	4.89374917737397\\
0.75	4.88209400684011\\
1	4.86253808460064\\
1.5	4.79294339317627\\
2	4.70073018899399\\
2.5	4.53082430370568\\
3	4.5086309987332\\
4	4.28542742293606\\
6	3.82902698099748\\
};
\addlegendentry{\textsc{LoneSTAR} \cite{roberts_lonestar} (tuned)}

\addplot [color=uclaviolet, line width=2.0pt, mark=*, mark options={solid, uclaviolet}]
table[row sep=crcr]{%
0.1	4.9524 \\
0.25	4.9295\\
0.5	4.9030\\
0.6	4.8867\\
0.75	4.8633\\
1	4.8512\\
1.5	4.8081\\
2	4.7528\\
2.5	4.6888\\
3	4.5767\\
4	4.3761\\
6	3.8714\\
};
\addlegendentry{\lonestar-WB (tuned)}

\addplot [color=uclablue, line width=2.0pt, mark=*, mark options={solid, uclablue}]
table[row sep=crcr]{%
    0.1	4.94978960557574\\
    0.25	4.93534365172467\\
    0.5	4.90648088296246\\
    0.6	4.89777788307566\\
    0.75	4.87849242221734\\
    1	4.87279895925223\\
    1.5	4.84264766392995\\
    2	4.8109092405668\\
    2.5	4.80890612555746\\
    3	4.75491131898717\\
    4	4.65882532337533\\
    6	4.36471701251282\\
};
\addlegendentry{Proposed}

\addplot [color=uclabluedarkest, line width=2.0pt, mark=*, mark options={solid, uclabluedarkest}, forget plot]
table[row sep=crcr]{%
    0.1	5.70903104067571\\
    0.25	5.7091566146789\\
    0.5	5.70867588023164\\
    0.6	5.70871668009014\\
    0.75	5.70841375229558\\
    1	5.70778099015369\\
    1.5	5.70593261028651\\
    2	5.70317393234311\\
    2.5	5.71858324679514\\
    3	5.69603619301197\\
    4	5.68595995974421\\
    6	5.65644615091503\\
} node[below=0mm,pos=0.815] {\small Codebook Capacity};

\end{axis}
\end{tikzpicture}%
%
    \vspace{-0.25cm}%
    \caption{Average sum spectral efficiency as a function of bandwidth $B$.}
    \label{fig:R}
    \vspace{-0.25cm}
\end{figure}

To better understand exactly why the proposed codebooks outperform \lonestar \cite{roberts_lonestar} and \lonestar-WB, let us consider \figref{fig:variance}, where we plot the downlink coverage variance achieved by each as a function of frequency $f$ for a bandwidth of $B=6$~GHz.
This downlink coverage variance is the left-hand side of \eqref{eq:opt-1-F-gain}, scaled down by $\Nt^2\Mtx$, defined as $\hat{\sigma}^2_{\mathrm{tx}}(f) \triangleq {\normtwo{\Nt \cdot \vone - \diag{\mAtx(f)\mF}}^2}  / {(\Nt^2 \, \Mtx)}$,  where $\mAtx(f)$ is defined analogous to \eqref{eq:Atx} for general $f$ instead of $f_k$.
The downlink coverage variance $\hat{\sigma}^2_{\mathrm{tx}}(f)$ serves as a measure of the variability in the gain provided by the transmit codebook in its intended directions across frequency.
The solid lines of \figref{fig:variance} depict the coverage variance of each codebook design with $\sigmatxsq=\sigmarxsq$ tuned to maximize the sum spectral efficiency.
\lonestar and \lonestar-WB maximize spectral efficiency across a 6~GHz bandwidth when $\sigmatxsq=\sigmarxsq=-14.1$~dB.
The proposed design maximizes spectral efficiency when $\sigmatxsq=\sigmarxsq=-8.5$~dB.
The dashed lines depict the coverage variance when $\sigmatxsq=\sigmarxsq$ is chosen to align the coverage variances at the center frequency.
In achieving a given coverage variance at the center frequency, \lonestar and \lonestar-WB exhibit more frequency-selectivity than the proposed design.
This is because \lonestar and \lonestar-WB neglect beam squint by only constraining the coverage variance at the center frequency, whereas the proposed design attempts to maintain coverage across the entire bandwidth.
Recall that the proposed design constrains the \textit{worst-case} coverage variance across subcarriers, which occurs at the edges of the band.
Thus, in pursuit of reducing wideband self-interference, the proposed design is incentivized to increase the coverage variance of frequency components across the band to give it more flexibility in minimizing the objective. 
This yields a coverage variance which is more frequency-flat. 
In turn, the solutions produced by the proposed design are more suitable in maximizing spectral efficiency since they provide more consistent coverage than \lonestar and \lonestar-WB.
This confirms that creating effective codebooks for wideband full-duplex wireless systems involves accounting for both frequency-selective self-interference and beam squint.

\begin{figure}[t]
    \centering%
    \input{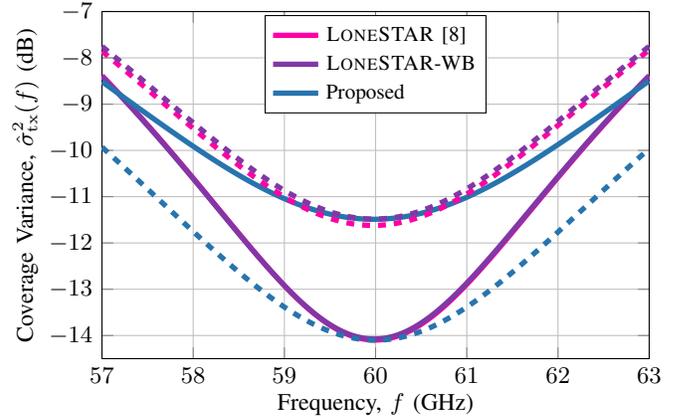}%
    \vspace{-0.25cm}%
    \caption{Downlink coverage variance $\hat{\sigma}^2_{\mathrm{tx}}(f)$ as a function of frequency. Our proposed design delivers higher spectral efficiency than \lonestar \cite{roberts_lonestar} and \lonestar-WB, since its codebooks exhibit more frequency-flat coverage.} 
    \label{fig:variance}
    \vspace{-0.25cm}
\end{figure}

\section{Conclusion} \label{sec:conclusion}

This paper has presented an analog beamforming codebook design for wideband full-duplex \mmwave systems.
Our proposed codebooks minimize self-interference across the entire band of interest and simultaneously deliver reliable coverage across that band.
The superior performance of our proposed design over existing narrowband codebooks for full-duplex \mmwave systems highlights the impacts of both beam squint and frequency-selective self-interference when operating over wide bandwidths.
Valuable future work may explore the applications of such codebooks in joint communication and sensing or the design of similar codebooks using machine learning.

\vspace{-0.1cm}






\bibliographystyle{bibtex/IEEEtran}
\bibliography{bibtex/IEEEabrv,refs}

\end{document}